\documentclass[11pt]{article}
\usepackage{fullpage,hyperref,amssymb,amsfonts,amsthm}

\usepackage{graphicx}
\usepackage{epsfig}
\usepackage{algorithm}
\usepackage{algpseudocode}

\newcommand{\<}{\langle}
\renewcommand{\>}{\rangle}
\newcommand{\C}{\mathbb{C}}

\newcommand{\cH}{\mathcal{H}}

\newcommand{\cN}{\mathcal{N}}

\newcommand{\ii}{\mathbb{I}} 

\newcommand{\be}{\begin{equation}}
\newcommand{\ee}{\end{equation}}
\newcommand{\bq}{\begin{eqnarray}}
\newcommand{\eq}{\end{eqnarray}}
\newcommand{\ba}{\begin{align}}
\newcommand{\ea}{\end{align}}

\newcommand{\ket}[1]{\left | \, #1 \right\rangle}
\newcommand{\bra}[1]{\left \langle #1 \, \right |}

\newcommand{\proj}[1]{\ket{#1}\bra{#1}}

\newcommand{\Sp}{\,\,\,\,\,\,}

\newcommand{\per}{\mathrm{per}}

\newtheorem{definition}{Definition}
\newtheorem{theorem}{Theorem}

\newtheorem{lemma}{Lemma}
\newtheorem{cor}{Corollary}

\begin{document}
\title{Efficient Computation of the Permanent of Block Factorizable Matrices}
\author{Kristan Temme\thanks{Center for Theoretical Physics, Massachusetts Institute of Technology, Boston, USA; \texttt{kptemme@mit.edu}} \quad\quad Pawel Wocjan\thanks{Mathematics Department and Center for Theoretical Physics, Massachusetts Institute of Technology, Boston, USA; on sabbatical leave from Department of Electrical Engineering and Computer Science, University of Central Florida, Orlando, USA; \texttt{wocjan@eecs.ucf.edu}}}
\date{August 31, 2012}

\maketitle

%
%
	
\begin{abstract}
We present an efficient algorithm for computing the permanent for matrices of size $N$ that can written as a product of $L$
block diagonal matrices with blocks of size at most $2$. For fixed $L$, the time and space resources scale linearly in $N$, with 
a prefactor that scales exponentially in $L$. This class of matrices contains banded matrices with banded inverse. We show that such a factorization into a product of block diagonal matrices gives rise to a circuit acting on a Hilbert space with a tensor product structure and that the permanent is equal to the transition amplitude of this circuit and a product basis state.  In this correspondence, a block diagonal matrix gives rise to one layer of the circuit, where each block to a gate acting either on a single tensor component or on two adjacent tensor components.  This observation allows us to adopt matrix product states, a computational method from condensed matter physics and quantum information theory used to simulate quantum systems, to evaluate the transition amplitude.
\end{abstract}


\section{Introduction}

For an arbitrary matrix $A\in M_{N\times N}(\C)$, its permanent is defined by
\begin{equation}
\per(A) = \sum_{\pi\in S_N} \prod_{i=1}^N A_{i,\pi(i)},
\end{equation}
where $S_N$ denotes the permutation group of the set $\{1,\ldots,N\}$.
The fastest known exact algorithm for computing the permanent of general matrices is due to Ryser \cite{Ryser63}, with running time $O(N 2^N)$.
Valiant showed that computing the permanent is \#P-hard \cite{Valiant79}.  Therefore, it is unlikely that there exists an efficient algorithm and attention has been given to approximating the permanent or computing it exactly only for restricted classes of matrices with special structure.  

We start by reviewing the approximation results.  For an arbitrary matrix $A\in M_{N\times N}(\C)$, Gurvits  \cite{Gurvits05} showed how to efficiently obtain an additive approximation $\tilde{p}\in \C$ such that
\[
| \tilde{p} - \mathrm{per}(A) | \le \varepsilon \cdot \|A\|^N
\]
holds with high probability, where $\|A\|$ is the matrix norm given by the largest singular eigenvalue of $A$.  The running time scales polynomially in $N$ and $1/\varepsilon$.

For an arbitrary matrix $A$ having only non-negative numbers as entries, Jerrum and Sinclair \cite{JS89} showed that it is possible to obtain a multiplicative approximation of the permanent.  More precisely, for $\varepsilon>0$, it is possible to obtain $\tilde{p}$ such that
\[
| \tilde{p} - \mathrm{per}(A) | \le \varepsilon \cdot \mathrm{per}(A).
\]
The running time scales polynomially in $n$ and $1/\varepsilon$.  A significantly faster algorithm for this problem was presented in \cite{BSVV08}.

In this article, we present a polynomial time algorithm for exactly computing the permanent for a new class of matrices.
We give a brief overview of some classes of matrices for which algorithms are known whose running time is polynomial or faster than that of Ryser's algorithm.

For any matrix with at most $c N$ nonzero entries, the permanent can be computed in time $O^*((2 - \varepsilon)^N)$ time, with 
$\varepsilon$ depending on $c$ \cite{SW05}.

Notable classes of matrices with special structure are those of circulant and Toeplitz matrices.  Some explicit solutions or recurrence relations for some $(0,1)$ circulant matrices are given in \cite{KP69,MSS69,Minc85,Minc87}.  Other $(0,1)$ circulant and very sparse Toeplitz matrices are discussed in \cite{CCR97}. 

Efficient algorithms can also be found for particular classes of banded matrices.  
A matrix $A$ is banded with bandwidth $w$ whenever $a_{ij}=0$ for $|i-j|>w$.  A banded matrix with bandwidth $w=1$ is called a tridiagonal matrix.  For a tridiagonal matrix, the permanent can be computed efficiently with the help of a recurrence relation \cite{daTonseca10}.  For a Toeplitz matrix $A$ of bandwith $w$, there is an algorithm that computes $\per(A)$ in time $O\left({2w \choose w}^3 \log N\right)$ \cite{Schwartz09}.

A further class is that of matrices whose permanent may be computed by the determinant of a matrix of the same size. This was established in the work of Temperley and Fisher \cite{TM60} and in the work of Kasteleyn \cite{Kasteleyn61,Kasteleyn67}, which were motivated by a problem in statistical mechanics.

The key ideas underlying the algorithm presented here are motivated by methods used in quantum information theory. We express the  permanent as the transition amplitude in a quantum circuit and apply matrix product states to evaluate this amplitude. A further example where matrix product states has been applied successfully to problems outside the field of quantum information theory is given in \cite{CM12}. The authors present an improved algorithm for counting the number of satisfying inputs of certain classical circuits.

We obtain an efficient algorithm for the new class of matrices of block factorizable matrices:

%
%

\bigskip
\noindent {\bf Main result:}
Assume that for $A \in M_{N\times N}(\C)$ we have $A=F_1 F_2 \cdots F_L$, where each of the factors $F_i$ is block diagonal with 2 by 2 or 1 by 1 blocks.  Given such factorization, the permanent of $A$ can be computed in time $O(N 2^{3L^2})$ and space $O(N 2^{2L^2})$.

\bigskip
Observe that any matrix $A$ that posses such factorization must be banded with bandwidth $L$. Further, if $A$ is invertible, then the inverse $A^{-1}$ must also be banded with bandwidth $L$.  


The problem of decomposing banded matrices with banded inverses into block diagonal matrices was studied by Strang \cite{Strang11_LPU,Strang11}.  The results therein show that any invertible matrix $A$ can be factored into $L=2 w(2w+1)$ such factors provided that $A$ and $A^{-1}$ are both banded with bandwidth $w$.
We later describe briefly how such decomposition can be computed efficiently. 

A problem which makes the computation of the permanent such a difficult task is the absence of simplification rules present for the determinant. 
The determinant already by construction obeys a multiplication rule, i.e. $\det(AB) = \det(A)\det(B)$. This rule does not hold for 
the permanent of a matrix product.  However, the permanent obey the following simplification rules involving permutation matrices $P$ and diagonal 
matrices $D$
\be
  \per(P  A  D) = \per(P) \; \per(A) \; \per(D) \Sp \mbox{and} \Sp \per(P A P') = \per(A).
\ee
These rules allow us to use the Cuthill--McKee algorithm to reduce the bandwidth of the matrix $A$ \cite{CutKee69} without changing its permanent.  Recall that the Cuthill--McKee algorithm suitably permutes the rows and columns of the sparse matrix to reduce its bandwidth.
 
%
%

\section{Permanent as the transition amplitude in a quantum circuit}

\subsection{Multivariate polynomials and the permanent}

We now derive a new alternative expression for the permanent making use of multivariate polynomials.  The importance of this expression is that our algorithm can be used to directly compute this expression.   The idea for this expression is motivated by the appearance of the permanent in quantum optics.  

Let $R=\C[X_1,\ldots,X_N]$ be the polynomial ring in the $N$ variables $X_1,\ldots,X_N$ over the field $\C$ of complex numbers. The variable $X_i$ corresponds to the creation operator $a_i^\dagger$ of a photon in the $i$th mode.  It is well-known that the permanent of a unitary matrix $U$ can be expressed as the amplitude $\<\hat{1}|\Phi_U|\hat{1}\>$, where $|\hat{1}\>$ is the Fock state in which there is exactly one photon in each of the $N$ modes, that is, $|\hat{1}\>=a_1^\dagger\cdots a_N^\dagger|\mathrm{vac}\>$ and $\Phi_U$ is the unitary describing the transformation on the Fock space induced by the transformation $(a_1^\dagger,\ldots,a_N^\dagger)\mapsto (a_1^\dagger,\ldots,a_N^\dagger) U$ of the creation operators.  This was observed in \cite{Scheel04,AA10}.

We abstract from the specific details pertaining to quantum optics and rely solely on multivariate polynomials to derive the new expression.  
This approach is more straightforward and better suited to obtaining and analyzing our algorithm.

\begin{definition}\label{def:phi}
Let $R=\C[X_1,\ldots,X_N]$. To an arbitrary matrix $A\in M_{N\times N}(\C)$, we associate the unique map $\Phi_A : R \rightarrow R$ by defining its action on the variables $X_1,\ldots,X_N$ by
\[
\Phi_A(X_i)=\sum_{j=1}^N A_{ij} X_j \quad\mbox{for $i=1,\ldots,N$}
\]
and lifting it to all polynomials by requiring that it acts additively and multiplicatively with respect to polynomial addition and multiplication, respectively, that is,
\[
\Phi_A(p +      q) = \Phi_A(p)+\Phi_A(q) \quad\mbox{and}\quad
\Phi_A(p \cdot q) = \Phi_A(p) \cdot \Phi_A(q) 
\]
for all $p,q\in R$.\footnote{$\Phi_A$ is an endomorphism of the ring $R$.}
\end{definition}

Observe that $\Phi_A$ preserves the total degree of polynomials.  Moreover, homogeneous polynomials are mapped onto homogeneous polynomials.  To express this more formally, we introduce the grading of $R$ with respect to total degree $d\ge 0$:
\begin{eqnarray}
\cN^{(d)} & = & \left\{ n=(n_1,\ldots,n_N)\in\mathbb{N}^N \mid \sum_{i=1}^N n_i = d \right\}, \\
R^{(d)} & = & \sum_{n\in\cN^{(d)}} c_n X^n,
\end{eqnarray}
where we use the shorthand notation $X^n=X_1^{n_1} X_2^{n_2} \cdots X_N^{n_N}$.  We have $\Phi_A(R^{(d)})\subseteq R^{(d)}$ for $d\ge 0$, with equality iff $A$ is invertible.

%
%

\begin{lemma}
The permanent of $A$ is related to $\Phi_A$ as follows:
\begin{equation}
\Phi_A(X_1 X_2 \cdots X_N) = \per(A) X_1 X_2 \cdots X_N + \sum_{n \in \cN^{(N)}\setminus(1,1,\ldots,1)} 
c_n X^n.
\end{equation}
\end{lemma}


\begin{proof}
Consider the homogeneous polynomial 
\[
\Phi_A(X_1 X_2 \cdots X_N) = 
\left( \sum_{j_1=1}^N A_{1,j_1} X_{j_1} \right) 
\left( \sum_{j_2=1}^N A_{2,j_2} X_{j_2} \right) \cdots 
\left( \sum_{j_N=1}^N A_{1,j_N} X_{j_N} \right).
\]
By inspection we see that the coefficient of $X_1 X_2 \cdots X_N$ in $\Phi_A(X_1 X_2 \cdots X_N)$ is equal to $\per(A)$.
\end{proof}

%
%

\begin{lemma}
For $B,C\in M_{N\times N}(\C)$, we have
\[
\Phi_{B C} = \Phi_{C} \circ \Phi_{B}.
\]
\end{lemma}
\begin{proof}
Assume $A=BC$ so that $A_{ij}=\sum_{k=1}^N B_{ik} C_{kj}$ for $i,j=1,\ldots,N$. We have
\begin{eqnarray*}
\Phi_A(X_i) 
& = & 
\sum_{j=1}^N A_{ij} X_j =
\sum_{j,k=1}^N B_{ik} C_{kj} X_j \\
& = & 
\sum_{k=1}^N B_{ik} \sum_{j=1}^n C_{kj} X_j =
\sum_{k=1}^N B_{ik} \Phi_C(X_k) \\
& = & 
\Phi_C\Big ( \sum_{k=1}^N B_{ik} X_k \Big) =
\Phi_C \big (\Phi_B(X_i) \big) \\
& = &
(\Phi_{C} \circ \Phi_{B})(X_i).
\end{eqnarray*}
We carried out this calculation explicitly to emphasize that the map $A\mapsto \Phi_A$ reverses the order, that is, $\Phi_{BC} = \Phi_C \circ \Phi_B$ and not $\Phi_B \circ \Phi_A$, as one might expect.
\end{proof}

\subsection{Factorization of $A$ and the corresponding quantum circuit}

We now show that the permanent can also be expressed as the transition amplitude of of a circuit acting on a Hilbert space with a tensor product state.  This allows us to use the method of matrix product states to compute the permanent.  To this end, we associate the Hilbert space 
\[
\cH = \mathrm{span}\{ |n\> \mid n \in \cN^{(N)} \}
\]
 to the set $R^{(N)}$ of homogeneous polynomials of total degree $N$.  
For $n\in\cN^{(N)}$, the monomial $X^n$ corresponds to the basis state $|n\>$.  
By construction, the map $\Phi_A$ is a linear map on this Hilbert space.  We use $|\hat{1}\>$ to denote the state $|1\>\otimes|1\>\otimes\cdots\otimes |1\>$, which corresponds to the monomial $X_1 X_2 \cdots X_N$.

It is important for our algorithm that the Hilbert space $\cH$ can be embedded into the tensor product Hilbert space 
\begin{equation}\label{eq:tensorProductH}
\big(\C^{N+1}\big)^{\otimes N} = \mathrm{span}\Big\{ |n_1\> \otimes \cdots \otimes |n_N\> \mid 0 \le n_i \le N \mbox{ for $i=1,\ldots,N$}\Big\}. 
\end{equation}
Since $\big(\C^{N+1}\big)^{\otimes N}$ has a tensor product structure, we can naturally consider circuits with the obvious notion of the two-qudit and single qudit gates.  We restrict the two qudit gates to act on adjacent qudits only.  

\begin{lemma}\label{lem:circuit}
Let $F\in M_{N\times N}(\C)$.  If $F$ has the form
\[
\left(
\begin{array}{ccccccc}
1 & \\
  & \ddots \\
  &  & 1 \\
  &  & & a \\
  &  & & & 1 \\
  & & & & & \ddots \\
  & & & & & & 1
\end{array}
\right)
\]
where the $1\times 1$ block $(a)$ acts on the subspace spanned by $|n_k\>$ of $\C^{(N+1)}$, then $\Phi_F$ is a single qudit gate
\[
\underbrace{\ii \otimes \cdots \otimes \ii}_{k-1} \otimes G'_k \otimes \underbrace{\ii \otimes \cdots \otimes \ii}_{N-k},
\] 
where $G'_k$ acts as diagonal matrix on $\C^{N+1}$.

If $F$ has the form
\[
\left(
\begin{array}{cccccccc}
1 &         \\
  &  \ddots \\
  &          & 1 \\
  &          &    & a & b \\
  &          &    & c & d \\
  &          &    &    &   & 1 \\
  &          &    &    &   &   & \ddots \\
  &          &    &    &   &   &           & 1
\end{array}
\right),
\]
where the 2 by 2 block $\left( \begin{array}{cc} a & b \\ c & d \end{array} \right)$ acts on the two qudit  subspace spanned by $|n_k\> \otimes |n_{k+1}\>$ in  $\big(\C^{N+1})^{\otimes 2}$, then $\Phi_F$ is a two qudit gate
\[
\underbrace{\ii \otimes \cdots \otimes \ii}_{k-1} \otimes G_{k,k+1} \otimes \underbrace{\ii \otimes \cdots \otimes \ii}_{N-k-1}.
\]
where $G_{k,k+1}$ acts on the tensor components $k$ and $(k+1)$ of $(\C^{N+1})^{\otimes N}$ 
and is block-diagonal with respect to the subspaces 
\[
\mathcal{S}(m) = \mathrm{span}\{ |n'\> \otimes |n''\> \mid n' + n'' = m\} \subset \C^{N+1}\otimes \C^{N+1}.
\]
\end{lemma}
\begin{proof}
Let us first assume, that $F$ only has a single $1\times 1$ block given by $a$ in the $k$'th entry. According to definition \ref{def:phi}, the map $\Phi_F$ acts as $\Phi_F(X_1^{n_1}X_2^{n_2} \ldots X_k^{n_k} \ldots X_N^{n_N}) =  a^{n_k} X_1^{n_1}X_2^{n_2} \ldots X_k^{n_k} \ldots X_N^{n_N}$ on the monomials. Due to the linearity of $\Phi_F$ this can be extended to the full space $(\C^{N+1})^{\otimes N}$. 
The action of the map $\Phi_F$ can be described by  $\ii^{\otimes k-1} \otimes G'_k \otimes \ii^{ \otimes n-k}$, with 
\be
	G'_k  = \sum_{n_k=0}^N a^{n_k} \proj{n_k}. 
\ee
This gate only acts on a single site. 

Furthermore, let us consider the action of a $2\times 2$ block $\left( \begin{array}{cc} a & b\\ c & d \end{array} \right)$, on the variables $X_k$ and $X_{k+1}$. If we assume that the $F$ acts trivially on the remaining variables, we have that
\bq
	\Phi_F(X_1^{n_1} \ldots X_k^{n_k} X_{k+1}^{n_{k+1}}\ldots X_N^{n_N}) =  \sum_{r_k = 0}^{n_k} \sum_{r_{k+1} = 0}^{n_{k+1}} 
	{n_k \choose r_k }{n_{k+1} \choose r_{k+1} } a^{r_k}b^{n_k - r_k}c^{n_{k+1} - r_{k+1}}d^{r_{k+1}} \times \nonumber \\
	X_1^{n_1} \ldots X_k^{r_k + n_{k+1} - r_{k+1} } X_{k+1}^{n_k - r_k + r_{k+1}}\ldots X_N^{n_N}.
\eq
Again, by making use of the linearity of $\Phi_F$, we can extend the action of this map to the full Hilbert space. This map  
acts on the state as the tensor product $\ii^{\otimes k-1} \otimes G_{k,k+1} \otimes \ii^{ \otimes n- k -1}$. For completeness, we state 
the gate $G_{k,k+1}$ explicitly,  

\bq\label{fugly}
G_{k,k+1} =  \sum_{n = 0}^{ 2 N} \sum_{m = 0}^{n} \sum_{r = 0}^{m} \sum_{s = 0}^{n-m} 
{m \choose r }{{n - m} \choose s } a^{r}b^{m-r}c^{n-m -s}d^{s} \times \nonumber \\  \ket{r + n - m -s}\ket{m - r + s}\bra{m}\bra{n-m}.
\eq

It follows that this gate has to preserve the total degree of the variables $X_k$ and $X_{k+1}$,
since by definition $\Phi_F(X_k^{n_k}X_{k+1}^{n_{k+1}}) =  \Phi_F(X_k)^{n_k}\Phi_F(X_{k+1})^{n_{k+1}}$. This leads to 
a block diagonal structure of the gate $G_{k,k+1}$. Each block is labeled by the total degree, or particle number, 
$n_{k} + n_{k+1}$. Note, that this is why we only  consider the sum of $n = 0,\ldots,2N$ in the representation in equation (\ref{fugly}). 
\end{proof}

\begin{cor}\label{cor:quantumCircuit}
A decomposition of $A$ of the form $A=F_1 F_2 \cdots F_L$, where each factor $F_i$ is block-diagonal with blocks of size 1 by 1 or 2 by 2, leads to a decomposition of $\Phi_A=\Phi_{F_L} \cdots \Phi_{F_2} \Phi_{F_1}$ as a quantum circuit of depth $L$.  The blocks of each factor $F_i$ correspond to gates that can be executed in parallel, that is, each $F_i$ corresponds to one layer $\Phi_{F_i}$.
\end{cor}

\subsection{Decomposition of banded matrices with banded inverses}
 
We now show that the class of block factorizable matrices contains banded matrices with banded inverses, which were studied in the literature \cite{Strang11_LPU,Strang11}.
 
Ideally, we would like to have sufficient and necessary conditions characterizing matrices block factorizable matrices $A$, that is, those having decompositions of the form
\[
A = F_1 F_2 \cdots F_L,
\]
where each $F_i$ is block-diagonal with 2 by 2 or 1 by 1 blocks.

It is easy to see that $A$ must be banded with bandwidth $L$ to possess such decomposition.  Furthermore, if $A$ is invertible then $A^{-1}$ must also be banded with bandwidth $L$ since
\[
A^{-1} = F_L^{-1} \cdots F_2^{-1} F_1^{-1}.
\]
The results on banded matrices with banded matrices due to Strang lead to the following theorem.

\begin{theorem}
Suppose $A$ and $A^{-1}$ have bandwidth $w$.  Then $A$ is a product $F_1\cdots F_L$ of block-diagonal matrices $F_i$ with 2 by 2 or 1 by 1 blocks and $L\le 2 w(2w+1)$.
\end{theorem}

The proof of consists of two steps:
\begin{itemize}
\item Factor $A$ into $BC$ with diagonal blocks of size $w, 2w, 2w,\ldots$ for $B$ and $2w,2w,\ldots,2w$ for $C$. The shift between the two sets of blocks means that $A=BC$ need not be block diagonal.
\item Break $B$ and $C$ separately into factors $F$ with blocks of size $2$ or $1$.  
\end{itemize}
The method for obtaining the decomposition $A=BC$ in the first step is described in \cite[Section 2]{Strang11_LPU}.  The special case of orthogonal matrices is treated in \cite[Theorem 2.1]{Strang11}.  

For the second step, we use the result that any matrix of size $n$ by $n$ can be decomposed into a product of $n(n+1)/2$ matrices such that each of these matrices contains only one 2 by 2 block.  This is described in \cite[Lemma 2.2]{Strang11}.  Applying this result independently to each block of $B$ and $C$ we conclude that $B$ and $C$ can be decomposed into a product of $w(2w+1)$ block diagonal matrices with 2 by 2 and 1 by 1 blocks.  We see that $L\le 2w(2w+1)$.

%
%

\section{Evaluation of circuits based on matrix product states}

A naive evaluation of the circuit for the permanent would immediately yield a cost which scales at least as unfavorably as 
$O((N+1)^N)$, when keeping track of all the components of the vector in the tensor product Hilbert space $\cH$ in eq.~(\ref{eq:tensorProductH}). However, for a circuit of small depth such as the one in Corollary~\ref{cor:quantumCircuit}, a significantly more compact description of the resulting state can be achieved.  A widely used description is the matrix product state (MPS) \cite{Perez,Vidal03,Verstraete08}. 
\begin{definition}
A matrix product state (MPS) on some tensor product space $\ket{\psi} \in (\C^{d+1})^{\otimes N}$ is a state written in the form
\be\label{def:MPS}
	\ket{\psi} = \sum_{n_1,\ldots,n_N = 0}^d B^{[1]}_{n_1}B^{[2]}_{n_2} \ldots B^{[N]}_{n_{N}} \ket{n_1\ldots n_N}.
\ee 
For $k=1,\ldots,N$, the matrices $B^{[k]}_{n_k} \in M_{D_{k-1} \times D_{k}}(\C)$ are rectangular with the convention $D_0=D_{N+1}=1$.  The first matrix $B^{[1]}_{n_1} \in M_{1\times D_1}(\C)$  is a row vector and the last matrix $B^{[N]}_{n_{N}} \in M_{D_{N-1}\times 1}(\C)$ is a column vector.  The physical index $n_k$ takes values in
$\{0,1,\ldots,d\}$. The size of the MPS tensor $B^{[k]}_{n_k} $ is described by the triple $(D_{k-1},D_k,d)$.
\end{definition}
We associate a matrix $B^{[k]}_{n_k}$ to each local basis element $\ket{n_k}$.  The coefficient for each basis element $\ket{n_1\ldots n_N}$ is obtained as the product of the corresponding matrices.  
It is often very convenient to visualize the MPS in terms of the pictorial representation given in Fig.~\ref{fig:MPS}.  In this representation, each closed line corresponds to an index that has to be contracted.  An open line represents a local basis element.  
\begin{figure}[h]
\begin{center}
\resizebox{0.815\linewidth}{!}
{\includegraphics{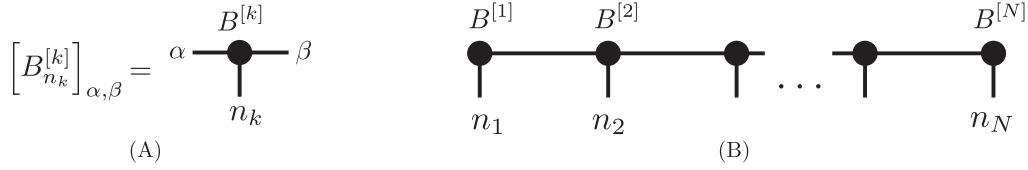}}
\caption{\label{fig:MPS} (A) The figure shows a graphical representation of the matrix product state tensor as defined in (\ref{def:MPS}). Each black dot corresponds to a matrix $B^{[k]}$ with elements $\left[B^{[k]}_{n_k}\right]_{\alpha, \beta}$. The three indices $(\alpha,\beta,n_k)$ are represented in terms of three black lines. (B) The coefficient for each basis state is obtained by contracting the virtual 
indices. This contraction of the $\alpha$ and $\beta$ indices is depicted as the closed center line. The full picture corresponds to
the matrix product $B^{[1]}_{n_1}B^{[2]}_{n_2} \ldots B^{[N]}_{n_N}$. The remaining indeces $\{n_k\}$ are left open.
This corresponds to the indeces of the tensor component which is supported on the local Hilbert space $\C^{d+1}$.}
\end{center}
\end{figure}

A subsequent application of the singular value decomposition between all bipartitions shows that any state in the Hilbert space $ (\C^{d+1})^{\otimes N}$ posesses such a matrix product state representation \cite{Vidal03}.  For the sake of completeness, we repeat the proof here.
\begin{lemma}
Any state $\ket{\psi} = \sum_{n_1,\ldots, n_N=0}^{d} c_{n_1\ldots n_N} \ket{n_1\ldots n_N} \in (\C^{d+1})^{\otimes N}$ possesses a matrix product representation of the form (\ref{def:MPS}).  The maximum matrix dimension $D_{max} = \max_k D_{k}$ is always bounded from above by  
$D_{max} \leq (d+1)^{\lfloor N/2 \rfloor}$.
\end{lemma}  
\begin{proof} 

Note that every state on a bipartite Hilbert space $\ket{\psi} \in {\cal H}_1 \otimes {\cal H}_2$ possesses a so called Schmidt decomposition. That is, we can write 
\be
\ket{\psi} = \sum_{i=0}^{d_1} \sum_{j=0}^{d_2} c_{ij} \ket{i j} \Sp =  \sum_{\alpha=0}^{\min(d_1,d_2)} \sigma_{\alpha}  \ket{\varphi_\alpha} \ket{\psi_\alpha}.
\ee
We made use of the canonical singular value decomposition (SVD) \cite{bhatia} for rectangular matrices. For the components of the 
matrix $c$ we have $c_{i,j} = \sum_{\alpha }U_{i ,\alpha} \, \sigma_\alpha \, \overline{V}_{j,\alpha}$ with $\sigma_\alpha\ge 0$.  Furthermore, the matrices $U = \sum_{a,b} U_{a,b} \ket{a}\bra{b}$ and $V = \sum_{a,b} V_{a,b} \ket{a}\bra{b}$ are unitary.  The so-called Schmidt vectors are given by $\ket{\varphi_\alpha} = \sum_{i} U_{i,\alpha} \ket{i}$ as well as $\ket{\psi_\alpha} = \sum_{j} \overline{V}_{j,\alpha} \ket{j}$.  They form orthonormal bases on the two sides of the bipartition due to the unitarity of $U$ and $V$.

To obtain the MPS representation, we proceed as follows:
\begin{enumerate}
\item
Consider the cut $(n_1),(n_2,\ldots,n_N)$ and perform a SVD of the matrix $c_{(n_1), (n_2\ldots n_N)}$ with respect to this partition. The resulting state can be expressed as  
\be
\ket{\psi} = \sum_{\alpha_1} \sigma_{\alpha_1}^{[1]} \ket{\varphi_{\alpha_1}^{[1]}} \ket{\psi_{\alpha_1}^{[2\ldots N]}} = \sum_{n_1, \alpha_1} U^{[1]}_{n_1,\alpha_1} \sigma_{\alpha_1}^{[1]} \ket{n_1} \ket{\psi_{\alpha_1}^{[2\ldots N]}}.
\ee
\item
The resulting Schmidt vector $\ket{\psi_{\alpha_1}^{[2\ldots N]}}$ (with support on the space $[2\ldots N]$) depends on the indices $\alpha_1,n_2,n_3,\ldots,n_N$. 
Introduce the next cut between $(\alpha_1,n_2),(n_3,\ldots,n_N)$  and perform a further SVD along this partition. This yields the representation
\be
\ket{\psi_{\alpha_1}^{[2\ldots N]}} = \sum_{n_2} \ket{n_2}\ket{\tau_{\alpha_1,n_2}^{[3\ldots N]}} =  \sum_{n_2} U^{[2]}_{(\alpha_1 n_2),\alpha_2 } \sigma_{\alpha_2}^{[2]}  \ket{n_2}\ket{\psi_{\alpha_2}^{[3\ldots N]}}.
\ee 
\item
We now proceed inductively by repeating step 2. After the final local space is reached, the resulting state is of the form
\be
\ket{\psi} = \sum_{\{n_k\}}\sum_{\{\alpha_k\}} U^{[1]}_{n_1,\alpha_1} \sigma^{[1]}_{\alpha_1}U^{[2]}_{(\alpha_1,n_2),\alpha_2} \sigma^{[2]}_{\alpha_2}\ldots U^{[N-1]}_{(\alpha_{N-2},n_{N-1}),\alpha_{N-1}} \sigma^{[N-1]}_{\alpha_{N-1}}\overline{V}^{[N]}_{n_N,\alpha_{N-1}} \ket{n_1 \ldots n_N}.
\ee 
\end{enumerate}
We obtain the desired MPS representation by defining the entries of the MPS tensors by setting 
\[
\left [B^{[k]}_{n_k}\right]_{\alpha_{k-1},\alpha_k} = U^{[k]}_{(\alpha_{k-1},n_k),\alpha_k} \sigma^{[k]}_{\alpha_k} 
\quad\mbox{and}\quad 
[B^{[N]}_{n_N}]_{\alpha_{N-1}} = \overline{V}^{[N]}_{n_N,\alpha_{N-1}}.
\]
Observe that the largest possible maximum dimension $D_k$ occurs for the cut at the center of the chain.  In this case, the dimensions of the respective Hilbert spaces are bounded from above by $(d+1)^{\lfloor N/2 \rfloor}$, which then corresponds to the largest possible rank of the matrix $S = \sum_{\alpha} \sigma_\alpha \proj{\alpha}$.  
\end{proof}

The matrix product state can be described with $\sum_{k = 1}^N (d+1) D_{k}D_{k+1} \leq O(N d D_{max}^2)$ parameters. This can be considerably less than $O((d+1)^N)$  when the matrix-dimensions $D_k$, also referred to as bond dimensions, are small. This scenario corresponds to a quantum state with little correlations among the qudits, which is said to possess a small amount of entanglement in quantum information theory.

Matrix product states and the application of gates to the states can conveniently be depicted as a tensor network (Fig.~\ref{fig:network}). The contraction of the entire tensor network corresponds to the evaluation of expectation values when all the lines are closed.
In this picture, the permanent can be understood as the contraction of the tensor network that is generated by the circuit decomposition of $\Phi_A$.

\begin{figure}[h]
\begin{center}
\resizebox{0.815\linewidth}{!}
{\includegraphics{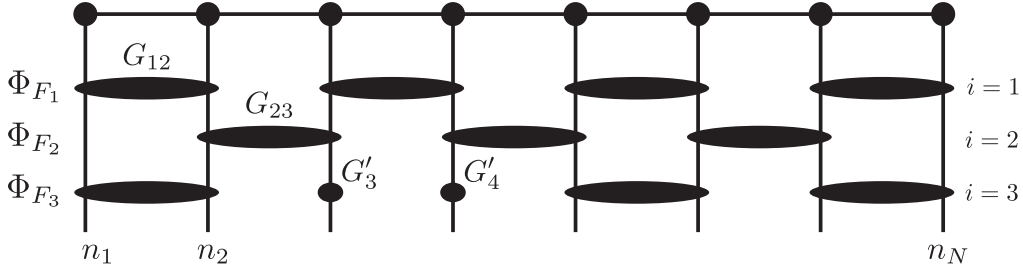}}
\caption{\label{fig:network} A circuit, such as the one considered in Corollary~\ref{cor:quantumCircuit}, can be understood as the contraction of a tensor network where the base corresponds to a matrix product state. We consider single qudit and two qudit operations acting on adjacent qudits. The former are represented by the gates $G'_k$ and the latter by the gates $G_{k,k+1}$. Each of these gates can be understood as local tensor in the tensor network. We have indicated the different layers $i=1,2,3\ldots$ of the network.  In each layer, the individual operations commute and the corresponding commuting gates can be performed in parallel.}
\end{center}
\end{figure}

During the evaluation of the circuit, it is desirable to retain the standard MPS form of the state to be computationally more efficient. A further result in \cite{Vidal03} shows at which cost local updates of a general MPS can be computed. A single qudit gate acting the site $k$ does not increase the Schmidt-rank of the MPS and can be applied to the local tensors $B^{[k]}$ individually. A two qudit gate acting the sites $k$ and $k+1$, however, does mix the neighboring tensors. The larger tensor resulting from this operation has to be decomposed again by a SVD, as depicted in Fig.~\ref{fig:transform}.

%
%
\begin{figure}[h]
\begin{center}
\resizebox{0.90\linewidth}{!}
{\includegraphics{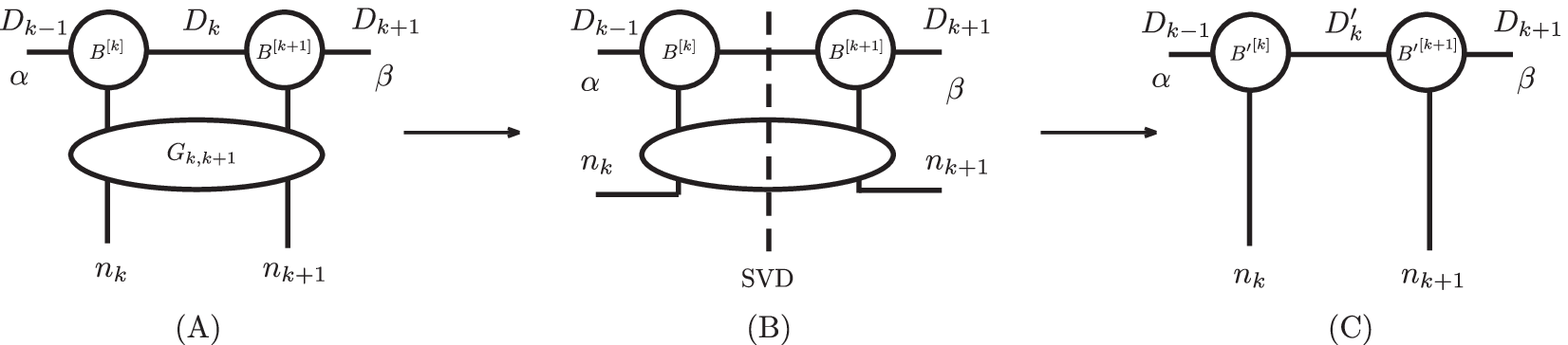}}
\caption{\label{fig:transform}The figure depicts the steps that are performed after an update of the MPS with a gate $G_{k,k+1}$ acting on two adjacent sites. The update proceeds as follows: 
(A) First we contract the MPS tensors $B^{[k]}$ and $B^{[k+1]}$ and apply the gate to the indices $n_k,n_{k+1}$. We are then left with a bigger tensor denoted by $M^{[k,k+1]}$.  This larger tensor has the components $[M^{[k,k+1]}_{n_k,n_{k+1}}]_{\alpha,\beta}$.  
(B) We then regroup the indices $(\alpha,n_k),(n_{k+1},\beta)$ and cast $M^{[k,k+1]}$ into matrix form $[M^{k,k+1}]_{(n_k,\alpha), (n_{k+1},\beta)}$ so that we can perform a SVD.  The resulting SVD yields the components $[M^{k,k+1}]_{(n_k,\alpha), (n_{k+1},\beta)} = \sum_\gamma U_{(n_k,\alpha),\gamma} \sigma_\gamma \overline{V}_{(n_{k+1},\beta),\gamma}$. 
(C) We then define the new MPS tensors ${B'}_{n_k}^{[k]}$ and ${B'}_{n_{k+1}}^{[{k+1}]}$ 
by setting their entries $\left[{B'}_{n_k}^{[k]}\right]_{\alpha,\gamma}=U_{(n_k,\alpha),\gamma} \sigma_\gamma$ and 
$\left[{B'}_{n_{k+1}}^{[k+1]}\right]_{\gamma,\beta}=\overline{V}_{(n_{k+1},\beta),\gamma}$, respectively.}
\end{center}
\end{figure}

%
%

\begin{lemma}\label{lem:update}
For a matrix product state $\ket{\psi}$ in the Hilbert space $(\C^{d+1})^{\otimes N}$, the following 
update operations can be performed so that $\ket{\psi}$ remain in matrix product form.
\begin{enumerate}
	\item An update corresponding to a single local gate $G'_k$ at site $k$ can be applied to the MPS $\ket{\psi}$ at cost
	         $O((d+1)^2 D_kD_{k+1})$.
	\item An update corresponding to a two qudit gate $G_{k,k+1}$ acting two adjacent sites ${k,k+1}$ can be applied at cost
	         $O(\min((d+1)D_{k-1},(d+1)D_{k+1})^3)$.
\end{enumerate} 
\end{lemma} 

\begin{proof}Let $\ket{\psi} = \sum_{\{n_k\}} B^{[1]}_{n_1}\ldots B^{[N]}_{n_N} \ket{n_1\ldots n_N}$.
\begin{enumerate}
\item   
The single qudit gate $G'_k$ acts as 
\be 
\ii^{\otimes(k-1)} \otimes G'_k \otimes \ii^{N-k} \ket{\psi} = \sum_{\{n_k\},n'_k} B^{[1]}_{n_1}\ldots \left[G'_{k}\right]_{n'_k,n_k}B^{[k]}_{n_k}\ldots B^{[N]}_{n_N} \ket{n_1\ldots n'_k \ldots n_N}.
\ee 
We see that it only modifies the $k$th MPS tensor $B^{[k]}$ and does not alter the other tensors. This update corresponds to a single tensor update that retains the overall structure of the matrix product state. The resulting tensor is given by $\left[{B'}^{[k]}_{n'_k}\right]_{\alpha,\beta} = \sum_{n_k} \left[G'_k\right]_{n'_k,n_k}\left[{B}^{[k]}_{n_k}\right]_{\alpha,\beta}$. The cost for this operation behaves as follows: for each pair $\alpha,\beta$ and each $n'_k$, we have to perform a vector multiplication at the cost $(d+1)$. The resulting overall cost is $O((d+1)^2D_k D_{k+1})$ since there are $D_k D_{k+1}$ pairs $\alpha$ and $\beta$ and $d+1$ indices $n'_k$.

\item
The two qudit gate $G_{k,k+1}$ acts on the state as
\bq 
&&\ii^{\otimes(k-1)} \otimes G_{k,k+1} \otimes \ii^{N-k-1} \ket{\psi} = \nonumber  \\ 
&& \sum_{\{n_k\},n'_k,n'_{k+1}} \!\! B^{[1]}_{n_1}\ldots \left[G_{k}\right]_{(n'_k,n'_{k+1}),(n_k,n_{k+1})}B^{[k]}_{n_k} B^{[k+1]}_{n_{k+1}} \ldots B^{[N]}_{n_N} \ket{n_1\ldots n'_k n'_{k+1} \ldots n_N}.
\eq 
We see that it combines the two, initially independent, tensors 
$B^{[k]}_{n_k} B^{[k+1]}_{n_{k+1}}$ 
to the larger tensor of the form
\be
\left[M^{[k,k+1]}_{n_k,n_{k+1}}\right]_{\alpha,\beta} = \sum_{n'_k,n'_{k+1}} \left[G_{k}\right]_{(n_k,n_{k+1}),(n'_k,n'_{k+1})}\left[B^{[k]}_{n'_k} B^{[k+1]}_{n'_{k+1}}\right]_{\alpha,\beta}.
\ee
This tensor has $(d+1)^2D_{k-1}D_{k+1}$ elements. The contraction of this larger tensor can be computed in time of the order 
$O((d+1)^4 D_{k-1}D_{k+1})$. 

After the contraction, we have to cast the tensors back into the original MPS form as explained in Fig.~\ref{fig:transform}. To achieve this, we regroup the indices into 
$(\alpha,n_k),(n_{k+1},\beta)$  and perform a SVD on the matrix $\left[M^{[k,k+1]}\right]_{(\alpha ,n_k),(n_{k+1},\beta)}$. 
The resulting matrix is $(d+1)D_{k-1} \times (d+1)D_{k+1}$ dimensional.  We are left with the decomposition 

\be
[M^{k,k+1}]_{(n_k,\alpha), (n_{k+1},\beta)} = \sum_{\gamma} U_{(n_k,\alpha),\gamma} \sigma_\gamma \overline{V}_{(n_{k+1},\beta),\gamma}.
\ee

The SVD can be performed in $O(\min((d+1)D_{k-1},(d+1)D_{k+1})^3)$ operations. Since we assume that $\min(D_{k-1},D_{k+1}) \geq d+1$.  The cost of the SVD outweighs the expense of constructing 
the tensor $M^{[k,k+1]}$ and we therefore state this bound as the total runtime of this update. After the SVD has been performed,
we define the new MPS tensors ${B'}_{n_k}^{[k]}$ and ${B'}_{n_{k+1}}^{[k+1]}$ by setting their entries 
\be \left[{B'}_{n_k}^{[k]}\right]_{\alpha,\gamma} = U_{(n_k,\alpha),\gamma} \sigma_\gamma \Sp \mbox{and} \Sp \left[{B'}_{n_{k+1}}^{[k+1]}\right]_{\gamma,\beta} = \overline{V}_{(n_{k+1},\beta),\gamma},
\ee
respectively.
\end{enumerate}
\end{proof}

In summary, the computational cost of contracting the tensor network in terms of time and space usage scales polynomially with the maximal bond dimension $D_{max}$. Therefore, an evaluation of the circuit is feasible when this dimension remains small. This  observation was made in\cite{Vidal03}. Certain quantum circuits can be simulated by classical means, provided that the entanglement is low during the entire computation.  We show in the next section that this is the case for the circuits in Corollary~\ref{cor:quantumCircuit}.

%
%

\section{Presentation of our algorithm and analysis of its time and space complexity}

To evaluate the circuit that corresponds to the permanent, we start the computation with the product state $\ket{\hat{1}} = \ket{11\ldots 11}$. For this state, we have $D_{max}=1$. However, this dimension grows with each application of a new layer of $\Phi(A)$.   In general, as we will see, the bond-dimension, and by that the resulting computational effort, can grow exponentially with the bandwidth of the matrix $A$. The following lemma states by how much each $D_k$ can increase when a single gate from Lemma~\ref{lem:circuit} is applied.

\begin{lemma}\label{lem:growdeg}
Let $\Phi_{F}$ denote a local gate as in Lemma \ref{lem:circuit}. After the application of $\Phi_F$ to the adjacent MPS tensors $B^{[k]}$ and $B^{[k+1]}$, with dimensions $(D_{k-1},D_{k},d)$ and $(D_{k},D_{k+1},d)$, respectively, the shared bond dimension $D_k$ behaves as follows:
\begin{enumerate}
	\item If $\Phi_F$ acts on a single site $k$, then the dimensions $D'_{k}$ remains unchanged. 
	\item If $\Phi_F$ acts on two adjacent sites $k,k+1$, then the new shared bond dimension $D'_k$ between these sites is is bounded by
	       ${D'}_k \le \min\left\{(d+1)D_{k-1},(d+1)D_{k+1}\right\}$.
\end{enumerate}
\end{lemma}

\begin{proof} The operations for the local updates are described in the proof of Lemma \ref{lem:update}.
\begin{enumerate}
\item
If the map $\Phi_F$ acts on a single site $k$, then the new MPS tensor is updated 
according to ${B'}^{[k]}_{n'_k} = \sum_{n_k} \left[\Phi_A\right]_{n_k,n_k} {B'}^{[k]}_{n'_k}$. Therefore, the triple $(D_{k-1},D_{k},d)$ remains unchanged and the MPS structure of the state is preserved.

\item
If the map $\Phi_F$ acts on the two adjacent sites $k$ and $k+1$, the tensor has to be updated according to $M^{k,k+1}_{n_k,n_{k+1}} = \sum_{n'_k,n'_{k+1}} \left[\Phi_A\right]_{(n_k,n_{k+1}),(n'_k,n'_{k+1})} B^{[k]}_{n'_k} B^{[k+1]}_{n'_{k+1}}$. This tensor corresponds to the one depicted in Fig.~\ref{fig:transform}.  After regrouping the indices we are left with a matrix of dimension $(d+1)D_{k-1} \times (d+1)D_{k+1}$. Therefore the singular value decomposition of the matrix $M^{k,k+1}$ can have at most the the rank $\min\left\{(d+1)D_{k-1},(d+1)D_{k+1}\right\}$, which gives the desired bound on $D'_k$.
\end{enumerate}
\end{proof}

Recall, that the gates $G_{k,k+1}$ and $G'_k$ are block diagonal, c.f. Lemma \ref{lem:circuit}. Therefore, as we will show now, the physical dimension is at most $d = 2^i$  after the application of $i$ layers. This fact can be exploited to achieve a linear scaling with
respect to $N$ in the time and space requirements. 

\begin{lemma}\label{lem:local-deg}
Assume we apply the first $i$ layers of the circuit $\Phi_A$, i.e. $\Phi_{F_i}\ldots\Phi_{F_1}$, corresponding to the decomposition $A = F_1\ldots F_i \ldots F_L$ (c.f. Fig.~\ref{fig:network}). If the initial state of the circuit is $\ket{\hat{1}} = \ket{11\ldots 1}$, then the local degree, i.e. particle number, $n_k$ is always bounded by $n_k \leq d = 2^i$ at every site $k$.
\end{lemma} 

\begin{proof}
The gates of the circuit  $G = \Phi_{F^k_i}$ can be constructed from $A = F_1\ldots F_i \ldots F_L$ as explained in Lemma \ref{lem:circuit}. Recall, that the gates $\Phi_{F^k_i}$ have the property that they preserve the total degree of the  monomials $X_k^n X_{k+1}^m$ and therefore are themselves block diagonal. Here each block corresponds to a particular local particle number. The local degree $d$ at a single site can increase when we have a gate that is supported on more than one site. Let us estimate by how much $d$ has to grow at each step. We consider the action of a set of commuting gates $\Phi_{F_i^k}$ on the basis states described by $X_1^{n_1} X_2^{n_2} \ldots X_N^{n_N}$. We write w.l.o.g.
\bq
&&G_{12}G_{34}\ldots G_{N-1,N}(X_1^{n_1} X_2^{n_2} \ldots X_N^{n_N})\nonumber \\ = && \Phi_{F_i^{12}}(X_1)^{n_1} \Phi_{F_i^{12}}(X_2)^{n_2} \ldots \Phi_{F_i^{N-1,N}}(X_{N-1})^{n_{N-1}}\Phi_{F_i^{N-1,N}}(X_N)^{n_N} \nonumber \\ =&& \left( \sum_r [F_i]_{1r}X_r \right)^{n_1}\left(\sum_r [F_i]_{2r}X_r\right)^{n_2} \ldots \left(\sum_r [F_i]_{Nr}X_r\right)^{n_N}.
\eq
By counting the degree, we see that the highest degree for two adjacent sites, e.g. $X_k$ and $X_{k+1}$, is always given by $n_k + n_{k+1}$. Therefore, we have that the local dimension increases as $d' = 2d$. We initially start in the state $\ket{\hat{1}}$ that corresponds to 
$X_1X_2\ldots X_N$. This state has a degree of $d = 1$ at each site. We proceed to apply layer by layer and therefore have that  
$d = 2^i$ after $i$ layers. 
\end{proof}

We are now ready to present the algorithm. The main procedure is written in pseudocode as stated in Table \ref{alg:Permanent}.
An informal description of the algorithm follows in the subsequent paragraph. 

\begin{table}[!hbtp]  
  \begin{algorithmic}[1]
    \Procedure{\textrm{(factorization of $A = F_1\ldots F_L$)}}{}
      \State prepare MPS $\{B^{[k]}\}_{k=1 \ldots N= \dim(A)}$ in $B^{[k]}_{n_k} = \delta_{1,n_k}$ for  $n_k \in \{0,1,d=2\}$;
      \For{$i=1$ to $L$}
          \For{$k =1$ block in $F_i = \bigoplus_{k=1} F_i^k$}
         	\State construct $G = \Phi_{F_i^k}$; 
		\State apply $G$ to $B^{[k]}B^{[k+1]}$;
		\State compute SVD and recast in standard MPS form;
		\State $k \leftarrow k + \dim(F_i^k)$;
        \EndFor
        		\If{$N \geq 2^{i+1}$}
        		\State increase local MPS dimension $d = 2^{i+1}$;
		\Else
		\State set local MPS dimension d = (N+1);
		\EndIf
		\State $i \leftarrow i+1$;
       \EndFor
   		\State \Return $\per(A) = B^{[1]}_1 B^{[2]}_1 \ldots B^{[N]}_1$;
    \EndProcedure
  \end{algorithmic}
  \caption{The algorithm computes the permanent of a matrix $AÊ\in M_{N\times N}(\C)$ given a factorization $A = F_1 \ldots F_L$ as input.}
\label{alg:Permanent}
\end{table}

We assume that we are given a decomposition of the matrix $A = F_1 \ldots F_L$, where $F_i = \bigoplus_{k=1} F_i^k$.
The blocks $F_i^k$ have size $2 \times 2$ and $1 \times 1$ as in Theorem~\ref{lem:circuit}. Each block corresponds to a local gate 
$G = \Phi_{F_i^k}$ in the circuit. The circuit starts with the initial state $\ket{\hat{1}}$. This state can be represented as a MPS with $D=1$ by choosing the $B^{[k]}_1 = 1$ for all $k = 1\ldots N$. All other components $B^{[k]}_{n_k}$ with $n_k > d=1$ of the MPS tensor are set to 
zero. Every layer corresponds to a single matrix $F_i$. The different blocks in the matrix commute and correspond to gates that can be applied in parallel. We update each MPS tensor $B^{[k]}$ as described in Lemma \ref{lem:update}.  This ensures that the state is still in matrix product form after the application of the new layer. Before we apply the next iteration, we have to increase the local dimension $d$ of the tensors to $d = 2^{\mbox{\#layer}}$. With each iteration the bond dimension $D_{max} $ will increase. Finally after the $L$ layers have been applied, we proceed to compute the scalar product with the state $\ket{\hat{1}}$. This can be done in time $O(N((d+1)D_{max})^3)$ and amounts to computing the matrix product $\per(A) = B^{[1]}_1 B^{[2]}_1 \ldots B^{[N]}_1$. The full estimate of the runtime, as well as the space requirements are stated in the following theorem.

\begin{theorem}\label{mainTheorem}
Assume that for $A \in M_{N\times N}(\C)$ we have $A=F_1 F_2 \cdots F_L$, where each of the factors $F_i$ is block diagonal with 
$2 \times 2$ or $1 \times 1$ blocks.  Given such factorization, the permanent of $A$ can be computed in time $O(N 2^{3L^2})$ and space $O(N 2^{2L^2})$.
\end{theorem}

\begin{proof}
Given the decomposition $A = F_1 \ldots F_L$, we construct the corresponding circuit as outlined in Lemma \ref{lem:circuit}.
The central figure of merit is the bond dimension $D_{max}$. Given the initial state $\ket{\hat{1}}$, this dimension is 
$D_{max} = 1$. However, with the application of each layer this dimension increases according to (c.f Lemma \ref{lem:growdeg})
${D'}_k = \min((d+1)D_{k-1},(d+1)D_{k+1})$. As we have outlined in the algorithm, the degree $d$ does not stay constant 
but grows at each iteration (c.f. Lemma \ref{lem:local-deg}), so that after $i$ layers, we have $d = 2^i$. Thus 
$ D_{max} \leq \Pi_{i=1}^L (2^i + 1) = O(2^{L^2}). $ The bound on $D_{max}$ can be used to give runtime and space bounds according to Lemma \ref{lem:update}. We need to store a total of $N$ MPS tensors each of size at most $O((2^{L^2},2^{L^2},2^L))$. We suppress the linear scaling of $L$ in the exponent. Therefore, in the worst case the space required scales at most as $O(N2^{2L^2})$. Furthermore, due to the runtime bound $O(\min((d+1)D_{k-1},(d+1)D_{k+1})^3)$ for the application of an individual gate, the total time is bounded by $O(N2^{3L^2})$.
\end{proof}
 
The runtime bound as stated in theorem \ref{mainTheorem} scales linearly in the matrix size $N$. The prefactor which is related to the bandwidth $w$ of the matrix via $w \leq L$ scales as $2^{L^2}$. However, observe that the highest dimension $d$ can never exceed $N$, which is the total degree of the full initial polynomial. The aforementioned bound is therefore only a good estimate for the worst case runtime, as long as $2^L \leq N$. When $2^L \geq N$ one can easily see, following the same arguments as in theorem \ref{mainTheorem}, that the runtime has to always be bounded by $O(N(N+1)^{3(L +1)})$. This bound is still polynomial in the matrix dimension. Furthermore, we would like to emphasize, that the stated bounds are only upper bounds for the worst case scenario for the algorithm. In practice it can occur that a large set of the singular values are zero after the gates $G_{k,k+1}$ has been applied. 

%
%

\section{Conclusions}

We have obtained a polynomial time algorithm for computing the permanent of block factorizable matrices.  The ideas behind this algorithm are the expression of the permanent as the transition amplitude in a quantum circuit and the application of matrix product states to evaluate this amplitude.

In the present work, we have limited ourselves to block factorizable matrices, that is, those with short factorizations into block diagonal matrices with blocks of size at most $2$.  These decompositions give rise to quantum circuit whose depths directly corresponds to the length of the factorization; the blocks of each factor correspond to quantum gates of one layer that can be executed in parallel.  To the best of our knowledge, this is the first time that a multiplicative decomposition has employed successfully to obtain an efficient algorithm for computing the permanent of matrices.

In a future publication, we will show how to rely on different techniques from condensed matter and quantum information theory such as matrix product operators \cite{Pirvu10} to obtain polynomial algorithms for computing the permanent of new classes of matrices.


\subsection*{Acknowledgements}
We would like to thank Eduardo Mucciolo and Frank Verstraete for insightful discussions.  
K.T. is grateful for the support from the Erwin Schr\"odinger fellowship, Austrian Science Fund (FWF): J 3219-N16.  P.W. gratefully acknowledges the support from the NSF CAREER Award CCF-0746600.  This work was supported in part by the National Science Foundation
Science and Technology Center for Science of Information, under
grant CCF-0939370.



\begin{thebibliography}{10}

\bibitem{AA10}
S.~Aaronson and  A.~Arkhipov, \emph{ The Computational Complexity of Linear Optics}, 2010; quant-ph/1011.3245, 
\texttt{http://arxiv.org/abs/1011.3245}

\bibitem{BSVV08}
I.~Bezakova, D.~Stefankovic, V.~V.~Vazirani and E.~Vigoda, \emph{Accelerating simulated annealing algorithm for the permanent and combinatorial counting problems}, SIAM Journal of Computing, 37(5), pp.~1429--1454, 2008.

\bibitem{bhatia}
R.~Bhatia, \emph{Matrix analysis}, Graduate text in mathematics, vol. 169, Springer.

\bibitem{CM12}
C.~Chamon and E.~R.~Mucciolo, \emph{Virtual parallel computing and a search algorithm using matrix product states}, Phys. Rev. Lett. 109, 030503, 2012.

\bibitem{CCR97}
B.~Codenotti, V.~Crespi, G.~Resta, On the permanent of certain $(0, 1)$ Toeplitz matrices, Linear Algebra Appl. 267, pp.~65Ð100, 1997.

\bibitem{CutKee69}
 E.~Cuthill and J.~McKee. \emph{Reducing the bandwidth of sparse symmetric matrices} In Proc. 24th Nat. Conf. ACM, pp.~157--172, 1969.

\bibitem{daTonseca10}
C.~M.~da Fonseca, \emph{The $\mu$-permanent of a tridiagonal matrix, orthogonal polynomials, and chain sequences}, Linear Algebra and its Applications, 432, pp.~1258--1266, 2010.

\bibitem{Gurvits05}
L.~Gurvits, \emph{On the complexity of mixed determinants and related problems}, Proc. Mathematical Foundations of Computer Science, pp.~447--458, 2005.

\bibitem{JS89}
M.~Jerrum and A.~Sinclair, \emph{Approximating the permanent}, SIAM Journal on Computing
18, pp.~1149--1178, 1989.

\bibitem{Kasteleyn61}
P.~W.~Kasteleyn, \emph{The statistics of dimers on a lattice: I. The number of dimer arrangements on a quadratic lattice}, Physica 27, pp.~1209--1225, 1961.

\bibitem{Kasteleyn67}
P.~W.~Kasteleyn, \emph{Graph theory and crystal physics}, in: F.~Harary (Ed.), Graph theory and theoretical physics, Academic Press, pp.~43--110, 1967.

\bibitem{KP69}
B.~W.~King and F.~D.~Parker, \emph{A Fibonacci matrix and the permanent function}, Fibonacci Quart. 7, pp.~539Ð544, 1969.

\bibitem{MSS69}
N.~Metropolis, M.~L.~Stein, P.~R.~Stein, \emph{Permanents of cyclic $(0, 1)$ matrices}, J. Combin. Theory Ser. B 7, pp.~291Ð321,  1969.

\bibitem{Minc85}
H.~Minc, \emph{Recurrence formulas for permanents of $(0, 1)$ circulants}, Linear Algebra Appl. 71, pp.~241Ð265, 1985.

\bibitem{Minc87}
H.~Minc, \emph{Permanental compounds and permanents of $(0, 1)$ circulants}, Linear Algebra Appl. 86, pp.~11Ð42, 1987.

\bibitem{Ryser63}
H.~J.~Ryser, \emph{Combinatorial Mathematics}, The Carus mathematical monographs, The Mathematical Association of America, 1963.

\bibitem{Scheel04}
S.~Scheel, \emph{Permanents in linear optical networks}, 2004; quant-ph/0406127, 
\texttt{http://arxiv.org/abs/quant-ph/0406127}

\bibitem{Schwartz09}
M.~Schwartz, \emph{Efficiently computing the permanent and Hafnian of some banded Toeplitz matrices}, Linear Algebra Appl. 430, pp.~1364--1374, 2009.
 
\bibitem{SW05}
R.~A.~Servedio and A.~Wan, \emph{Computing sparse permanents faster}, Information Processing
Letters 96(3), pp.~89--92, 2005.
 
\bibitem{Strang11}
G.~Strang, \emph{Groups of banded matrices with banded inverses}, Proc. American Mathematical Society 139(12), pp. 4255--4264, 2011.

\bibitem{Strang11_LPU}
G.~Strang, \emph{Banded matrices with banded matrices and $A=LPU$}, 2011.

\bibitem{TM60}
H.~N.~V.~Temperly and M.~E.~Fisher, \emph{Dimer problem in statistical mechanics -- an exact result}, 
Philosophical Magazine, 6(68), pp. 1061-1063, 1960.

\bibitem{Valiant79}
L.~G.~Valiant, \emph{The complexity of computing the permanent}, Theor. Comput. Sci. 8, pp.
189--201, 1979.

\bibitem{Perez}
D. Perez-Garcia, F. Verstraete, M. Wolf, and J. Cirac. \emph{Matrix product state representations}, 
Quantum Information \& Computation archive, 7(5), pp.~401--430, 2007.

\bibitem{Pirvu10}
B.~Pirvu, V.~Murg, J.~I.~ Cirac and F.~Verstraete. \emph{Matrix product operator representations}, New J. Phys. 12, pp.~025012, 2010.

\bibitem{Verstraete08} F.~Verstraete, V.~Murg and J.I.~Cirac \emph{ Matrix product states, projected entangled pair states, and variational renormalization group methods for quantum spin systems}, Advances in Physics 57(2), 2008.

\bibitem{Vidal03}
G.~Vidal, \emph{Efficient classical simulation of slightly entangled quantum computations}, Phys. Rev. Lett. 91, 147902, 2003.

\end{thebibliography}
\end{document}